\documentclass[a4paper, 11pt]{article}
\usepackage[top=1.0 in,bottom=1.0 in,left=1.0 in,right=1.0 in]{geometry}
\usepackage{amsmath,amsthm}
\usepackage{amssymb}
\usepackage[colorlinks,linkcolor=blue,citecolor=blue]{hyperref}
\usepackage{float}
\usepackage{cite}
\usepackage{epic}
\usepackage{tikz}
\usepackage{graphicx}
\usepackage{accents}
\usepackage{slashed}
\usepackage{mathrsfs} 
\usepackage{bm}

\newtheorem{theorem}{Theorem}[section]

\numberwithin{equation}{section}
\newtheorem{remark}{Remark}[section]

\allowdisplaybreaks

\title{Elliptic soliton solutions of lattice eigenfunction equations of the KdV type}

\author{Xing Li$^{1}$\footnote{Email: xing\_li@jsnu.edu.cn},
		~~ Ying Shi$^{2}$\footnote{Email: yingshi@zust.edu.cn},
		~~ Da-jun Zhang$^{3,4}$\footnote{Corresponding author. Email: djzhang@shu.edu.cn},
        ~~ Ruixiao Zhu$^{3,4}$
		\\ [5pt]
		{\small $^1$School of Mathematics and Statistics, Jiangsu Normal University, Xuzhou 221116, China} \\
		{\small $^{2}$School of Science, Zhejiang University of Science and Technology, Hangzhou 310023, China} \\
		{\small $^{3}$Department of Mathematics, Shanghai University, Shanghai 200444, China} \\
		{\small $^{4}$Newtouch Center for Mathematics of Shanghai University,  Shanghai 200444, China} 	}

\date{\today}

\begin{document}
\maketitle
\begin{abstract}
It is known that some integrable lattice equations can be formulated by the eigenfunctions
of the Lax triplet of the lattice potential Korteweg-de Vries (lpKdV) equation;
the eigenfunctions satisfy a discrete Abel formula in which the `integration' constant
defines the $\delta$ parameter in some equations in the Adler-Bobenko-Suris (ABS) classification.
In this paper, we present a full Darboux transformation for the Lax triplet of the lpKdV equation,
and eigenfunctions are formulated in terms of Casoratians.
When the eigenfunctions take the elementary discrete Lam\'e functions,
we get the elliptic $N$-soliton solution for the lpKdV equation after $N$-fold Darboux transformation.
We then determine the constant in the discrete Abel formula in this case.
The eigenfunctions are used to formulate elliptic $N$-soliton solutions for
various equations in the ABS list, such as  H3($\delta$), $Q1(\Delta)$, Q2 and Q3($\Delta$) equations.

\begin{description}
\item[Keywords:] eigenfunction equation; elliptic soliton; Casoratian;
lattice KdV; Darboux transformation
\item[PACS:] 02.30.Ik; 02.30.Jr; 05.45.Yv; 02.30.Gp
\end{description}
\end{abstract}

\section{Introduction}\label{sec:1}

In general, a (1+1)-dimensional integrable system can be formulated from the compatibility of
its Lax pair, whose solutions are known as eigenfunctions of the Lax pair.
By eigenfunction equation we mean the equation governed by eigenfunctions.
A well-known example is that the modified Korteweg-de Vries (mKdV) equation
is formulated by the eigenfunction of the Lax pair of the KdV equation via
the Cole--Hopf transformation.
Konopelchenko used to investigate various eigenfunction equations
related to the KdV equation \cite{Kon-RMP-1990}.
In a recent paper \cite{WZZZ-JPA-2024}, four lattice eigenfunction equations
were constructed from the Lax triplet (see eq.\eqref{eq:laxpair-H1})
of the lattice potential KdV (lpKdV) equation,
denoted as leKdV-I, leKdV-I($\delta$),  leKdV-II and leKdV-II($\delta$), respectively,
where $\delta$ is a parameter.
They correspond to the H3(0), H3($\delta$), Q3(0) and Q3($\delta$) equations
in the  Adler--Bobenko--Suris (ABS) classification \cite{ABS-CMP-2003}
for the 3-dimensionally consistent quadrilateral equations with extra constraints
(affine linear, D4 symmetry and tetrahedron property).
Since Q3($\delta$) is the top equation in the ABS list (except Q4),
in light of the known degeneration relation (cf.\cite{NAH-JPA-2009}),
the equations Q1(0), Q1($\delta$) and Q2 in the ABS list can also be formulated
by the eigenfunctions of the Lax triplet of the lpKdV equation \cite{WZZZ-JPA-2024}.
Thus, one can solve these equations by determining eigenfunctions.

The purpose of this paper is to present elliptic solitons in Casoratian form
for the lattice eigenfunction equations of the KdV type.
Here by ``elliptic solitons'' we mean the solitons living on an elliptic background,
which is different from the elliptic solitons introduced by J.-L. Verdier \cite{Ver-1988}
that are independent of dispersion relations (see a short review in Appendix C in \cite{LSZ-Non-2025}).
Such type of elliptic solitons have been derived for the KdV equation
as early as in 1970s by  Kuznetsov and Mikhailov
using the Gel'fand--Levitan--Marchenko (GLM) equation \cite{KM-JETP-1974}
and by Wahlquist using B\"acklund transformation \cite{Wahl-1976}.
Later, for the ABS equations (except Q4),
elliptic solitons in Cauchy matrix forms were constructed in 2010 by
Nijhoff and Atkinson through the Cauchy matrix approach \cite{NA-IMRN-2010}.
More recent progress includes
the establishment of the bilinear framework \cite{LZ-JNS-2022,LZ-2024},
the elliptic direct linearization scheme \cite{NSZ-CMP-2023,SZN-LMP-2026}
and the GLM equation \cite{LSZ-Non-2025},
and transformation and degeneration approaches based on the Sato theory
\cite{Kak-arx-2023,Nak-SIGMA-2026}.
There are also some results obtained by using Darboux transformation, e.g.\cite{WYSL-JMP-2025}.

In this paper, we will construct explicit expressions for the eigenfunctions
of the Lax triplet of the lpKdV equation by means of Darboux transformation.
Note that the discrete Crum theorem for the discrete Schr\"odinger spectral problem \eqref{eq:laxpair-H1-A}
has been established in \cite{LW-AMSSA-2006} and \cite{ZPZ-TMP-2020}.
However, the Darboux transformation for the whole Lax triplet \eqref{eq:laxpair-H1}
is still incomplete: Crum's theorem for \eqref{eq:laxpair-H1-C} is not verified so far.
We will prove it in this paper. This will allow us to have a Casoratian expression
for the eigenfunction associated with elliptic solitons.
After that, we will determine the constant in the discrete Abel formula.
Note that this constant plays a key role in defining the parameter $\delta$
in the leKdV equations and in H3($\delta$), Q1($\Delta$) and Q3($\Delta$).
Finally, we will present elliptic solitons in Casoratian forms
for Q3($\Delta$), Q3($0$), Q2, Q1($\Delta$), Q1($0$),
H3($\delta$)  and H3($0$).

The paper is organized as follows.
In Section \ref{sec:2}, we describe a full Darboux transformation for the
Lax triplet \eqref{eq:laxpair-H1}, getting Casoratian formulae for the
elliptic solitons of the lpKdV equation  as well as the eigenfunctions.
In Section \ref{sec:3}, we construct elliptic $N$-soliton solutions of
the lattice eigenfunction equations, together with determining the constant
in the discrete Abel formula.
In Section \ref{sec:4}, the Q1($\Delta$) and Q2 equations are treated via the
degenerating limits of the Q3($\Delta$) equation.
Explicit formulae for the elliptic solitons are obtained in terms of Casoratians.
Conclusions are provided in Section \ref{sec:5}.
Some determinant identities and properties of Weierstrass elliptic functions used in the paper
are collected in Appendix \ref{app:det-elliptic}.

\section{Darboux transformation and elliptic solitons for the lpKdV}\label{sec:2}

\subsection{Lax triplet}\label{sec:2-1}

In \cite{WZZZ-JPA-2024}, the lattice eigenfunction equations of the KdV type
were constructed from the following triplet
\begin{subequations}\label{eq:laxpair-H1}
\begin{align}
&\widetilde{\widetilde{\phi}}
+(\widetilde{\widetilde{u}}-u)\widetilde{\phi}
+\alpha^2\phi=\mathrm{\mathbf{k}}^2\phi ,\label{eq:laxpair-H1-A}\\
&\widehat{\widehat{\phi}}
+(\widehat{\widehat{u}}-u)\widehat{\phi}
+\beta^2\phi=\mathrm{\mathbf{k}}^2\phi,\label{eq:laxpair-H1-B}\\
&\widehat{\phi}=\widetilde{\phi}-(\widehat{u}-\widetilde{u})\phi.\label{eq:laxpair-H1-C}
\end{align}
\end{subequations}
Here $u=u(n,m)$ is a function of $(n,m)\in \mathbb{Z}^2$,
$\phi=\phi(u,\mathbf{k})$ is the eigenfunction,
$\alpha$ and $\beta$ are respectively the spacing parameters in $n$-direction and $m$-direction,
$\mathbf{k}$ is the spectral parameter.
For the tilde-hat notations, we denote
\begin{equation}
\widetilde{u}=u(n+1,m)=E_{[n]}u,~~ \widehat{u}=u(n,m+1)=E_{[m]}u, ~~
\widehat{\widetilde{u}}=u(n+1,m+1)=E_{[n]}E_{[m]}u,
\end{equation}
where $E_{[n]}$ and $E_{[m]}$ serve as the shift operators in
$n$-direction and $m$-direction, respectively.

The set \eqref{eq:laxpair-H1} can be viewed as a Lax triplet
for the lpKdV equation
 \begin{equation}\label{eq:H1}
 (u - \widehat{\widetilde{u}})(\widehat{u} - \widetilde{u}) = \alpha^2 - \beta^2.
 \end{equation}
In fact, shifting \eqref{eq:laxpair-H1-C} in the $n$-direction and replacing
$\widetilde{\widetilde{\phi}}$ using \eqref{eq:laxpair-H1-A} yield
\begin{equation}\label{eq:phi-th-from-n}
 \widehat{\widetilde\phi}
 =(\mathrm{\mathbf{k}}^2-\alpha^2)\phi+(u-\widehat{\widetilde u})\widetilde\phi.
\end{equation}
On the other hand, shifting \eqref{eq:laxpair-H1-C} in the $m$-direction and then using
\eqref{eq:phi-th-from-n} to replace $\widehat{\widetilde\phi}$, we obtain
\begin{equation}\label{eq:phi-hh-1}
 \widehat{\widehat\phi}
 =\left(\mathrm{\mathbf{k}}^2-\alpha^2
 +(\widehat{\widehat u}-\widehat{\widetilde u})
 (\widehat u-\widetilde u)\right)\phi+(u-\widehat{\widehat u})\widetilde\phi.
\end{equation}
Meanwhile, eliminating $\widehat\phi$ from \eqref{eq:laxpair-H1-B} and \eqref{eq:laxpair-H1-C} yields
\begin{equation}\label{eq:phi-hh-2}
 \widehat{\widehat\phi}
 =\left(\mathrm{\mathbf{k}}^2-\beta^2
 +(\widehat{\widehat u}-u)
 (\widehat u-\widetilde u)\right)\phi+(u-\widehat{\widehat u})\widetilde\phi.
\end{equation}
Then, equating the coefficients of $\phi$ in \eqref{eq:phi-hh-1} and \eqref{eq:phi-hh-2}
immediately gives rise to  the  lpKdV equation \eqref{eq:H1}.

\subsection{Darboux transformation}\label{sec:2-2}

Equation \eqref{eq:laxpair-H1-A} (or \eqref{eq:laxpair-H1-B})
can be considered as the discrete Schr\"odinger spectral problem.
The discrete Crum theorem for \eqref{eq:laxpair-H1-A}
has been proved in \cite{LW-AMSSA-2006} and \cite{ZPZ-TMP-2020}.
In this subsection, we show the full Darboux transformation for the whole Lax triplet \eqref{eq:laxpair-H1}.
The point is we need to prove the Darboux transformation of \eqref{eq:laxpair-H1-A} is valid
for \eqref{eq:laxpair-H1-B} and \eqref{eq:laxpair-H1-C} as well.
Consequently, we will get general Casoratian expressions for the potential function $u$
and eigenfunction $\phi$.

Let $u_0$ be an arbitrary seed potential in the triplet \eqref{eq:laxpair-H1},
and let $\phi(\mathbf{k})$ be an eigenfunction associated with $u_0$ and the
spectral parameter $\mathbf{k}$.
Choose $N$ distinct spectral parameters $\mathbf{k}_1^2,\ldots,\mathbf{k}_N^2$, and let  $\phi_i=\phi(n,m;\mathbf{k}_i)$ be the corresponding eigenfunctions.
In other words, for each $i$, the pair $(u_0,\phi_i)$ satisfies
the Lax triplet \eqref{eq:laxpair-H1} with $u=u_0$ and $\mathbf{k}^2=\mathbf{k}_i^2$.
Define
\begin{equation}\label{eq:r1-s1-def}
 r_1=-\frac{\widetilde{\phi}_1}{\phi_1},
 \qquad
 s_1=-\frac{\widehat{\phi}_1}{\phi_1}.
\end{equation}
From \eqref{eq:laxpair-H1-C} it is easy to see that
\begin{equation}\label{2.8}
\widetilde{u}_0-\widehat{u}_0=r_1-s_1,
\end{equation}
and consequently, from \eqref{eq:laxpair-H1-C} we have
\begin{equation}\label{2.9}
\widehat{\phi}-\widetilde{\phi}=(r_1-s_1)\phi.
\end{equation}
The one-step Darboux transformation for \eqref{eq:laxpair-H1-A} introduced in \cite{ZPZ-TMP-2020}
is composed of
\begin{subequations}\label{eq:one-step-DT-phi-u}
\begin{align}
 \phi[1] &:=\widetilde{\phi}+r_1\phi,\label{eq:one-step-DT-phi}\\
 u[1]&:=\widetilde{u}_0-r_1,\label{eq:one-step-DT-u}
\end{align}
\end{subequations}
which indicates that $\phi[1]$ and $u[1]$ satisfy \eqref{eq:laxpair-H1-A} as well, i.e.
\begin{equation}
\widetilde{\widetilde{\phi[1]}}
+(\widetilde{\widetilde{u[1]}}-u[1])\widetilde{\phi[1]}
+\alpha^2\phi[1]=\mathrm{\mathbf{k}}^2\phi[1],\label{eq:laxpair-H1-A-1}
\end{equation}
Obviously, by virtue of relations \eqref{2.8} and \eqref{2.9},
the transformations in \eqref{eq:one-step-DT-phi-u} can be equivalently defined as
\begin{subequations}\label{eq:one-step-DT-phi-u-2}
\begin{align}
 \phi[1] &:=\widehat{\phi}+s_1\phi,\label{eq:one-step-DT-phi-2}\\
 u[1]&:= \widehat{u}_0-s_1.\label{eq:one-step-DT-u-2}
\end{align}
\end{subequations}

In the following, we show that the Darboux transformations in \eqref{eq:one-step-DT-phi-u}
or \eqref{eq:one-step-DT-phi-u-2} are valid for \eqref{eq:laxpair-H1-B} and \eqref{eq:laxpair-H1-C}.
First, the  Lax triplet \eqref{eq:laxpair-H1} for $(u_0, \phi_1)$ reads
\begin{subequations}\label{eq:s1-r1-relation}
\begin{align}
&\bigl(\widetilde r_1-\widetilde{\widetilde u}_{0}+u_{0}\bigr)r_1=\mathbf{k}_1^2-\alpha^2,\label{eq:s1-r1-b}\\
&\bigl(\widehat s_1-\widehat{\widehat u}_{0}+u_{0}\bigr)s_1=\mathbf{k}_1^2-\beta^2, \label{eq:s1-r1-c}\\
& s_1-r_1=\widehat{u}_0-\widetilde{u}_0. \label{eq:s1-r1-a}
\end{align}
\end{subequations}
Corresponding to \eqref{eq:laxpair-H1-B}, we introduce
an operator
\begin{equation}
\mathcal L_{\beta}
 =E_{[m]}^2
 +(\widehat{\widehat u}_{0}-u_{0})E_{[m]}
 +\beta^2.
\end{equation}
By virtue of \eqref{eq:s1-r1-c},  the operator $\mathcal L_{\beta}-\mathbf{k}_1^2$ admits the factorization
\begin{equation}\label{eq:L-factorization-2}
\mathcal L_{\beta}-\mathbf{k}_1^2
=\left(E_{[m]}-\widehat s_1+\widehat{\widehat u}_{0}-u_{0}\right)
\left(E_{[m]}+s_1\right).
\end{equation}
Define a new operator by dressing
\begin{equation}\label{eq:L1-factorization-2}
\mathcal L_{\beta}[1]-\mathbf{k}_1^2
=\left(E_{[m]}+s_1\right)(\mathcal L_{\beta}-\mathbf{k}_1^2)\left(E_{[m]}+s_1\right)^{-1},
\end{equation}
i.e.
\begin{align}
\mathcal L_{\beta}[1]-\mathbf{k}_1^2
&=\left(E_{[m]}+s_1\right)
 \left(E_{[m]}-\widehat s_1+\widehat{\widehat u}_{0}-u_{0}\right) \nonumber\\
 &=E_{[m]}^2
 +(s_1-\widehat{\widehat s}_1
 +\widehat{\widehat{\widehat u}}_{0}-\widehat u_{0})E_{[m]}
 +\beta^2-\mathbf{k}_1^2,
\end{align}
which, in light of \eqref{eq:one-step-DT-phi-u-2}, gives rise to
\begin{equation}
\mathcal L_{\beta}[1]
 =E_{[m]}^2 +(\widehat{\widehat{u[1]}}-u[1])E_{[m]}+\beta^2
\end{equation}
On the other hand, acting the dressing from \eqref{eq:L1-factorization-2}
on $\phi[1] =(E_{[m]}+s_1)\phi$ immediately yields
\begin{equation}
 \qquad
\mathcal L_{\beta}[1]  \phi[1]=\mathbf{k}^2 \phi[1].
\end{equation}
This means \eqref{eq:one-step-DT-phi-u-2} provides a Darboux transformation for \eqref{eq:laxpair-H1-B}.
Note that the above dressing technique is the same as what has been used in \cite{ZPZ-TMP-2020}
for \eqref{eq:laxpair-H1-A},
while here we have made use of the equivalent transformation \eqref{eq:one-step-DT-phi-u-2}.

Now we consider the  relation \eqref{eq:laxpair-H1-C}.
An up-tilde shift of \eqref{eq:laxpair-H1-C} yields
\begin{equation*}
 \widehat{\widetilde{\phi}}_{1}
 =\widetilde{\widetilde{\phi}}_{1}
 -\left(\widehat{\widetilde u}_{0}-\widetilde{\widetilde u}_{0} \right)\widetilde\phi_{1}
 =r_1
 \left(\widetilde r_1+\widehat{\widetilde u}_{0}-\widetilde{\widetilde u}_{0} \right)\phi_{1}.
\end{equation*}
On the other hand, we have an alternative expression
\begin{equation*}
 \widehat{\widetilde{\phi}}_{1}= -\widehat r_1 \widehat \phi_1
 = \widehat r_1
 \left(r_1+\widehat u_{0}-\widetilde u_{0}\right)\phi_{1},
\end{equation*}
where \eqref{eq:s1-r1-a} has been used.
Thus we obtain a relation
\begin{equation}\label{eq:r1-compatible}
 \widehat r_1
 \left(\widehat u_{0}-\widetilde u_{0}\right)
 =
 r_1 \left(\widetilde r_1-\widehat r_1+
 \widehat{\widetilde u}_{0}-\widetilde{\widetilde u}_{0} \right).
\end{equation}
Combining it and the one-step potential \eqref{eq:one-step-DT-u} yields
\begin{equation}\label{eq:key-one}
 \widehat r_1(\widehat u_{0}-\widetilde u_{0})
 = r_1\bigl(\widehat{u[1]}-\widetilde{u[1]}\bigr).
\end{equation}
Now we come to check that the transformation \eqref{eq:one-step-DT-phi-u} works for
\eqref{eq:laxpair-H1-C}.
From \eqref{eq:one-step-DT-phi} and \eqref{eq:laxpair-H1-C} we have
\begin{align}
 \widehat{\phi[1]} &=\widehat{\widetilde\phi}+\widehat r_1 \widehat\phi \nonumber\\
 &=\widetilde{\widetilde\phi}
 - \left(\widehat{\widetilde u}_{0}-\widetilde{\widetilde u}_{0} \right)\widetilde\phi
 + \widehat r_1(\widetilde\phi-(\widehat u_{0}-\widetilde u_{0})\phi) \nonumber\\
 &=\widetilde{\widetilde\phi}
 + \left(\widehat r_1 -\widehat{\widetilde u}_{0}+\widetilde{\widetilde u}_{0} \right)\widetilde\phi
 -\bigl(\widehat{u[1]}-\widetilde{u[1]}\bigr)r_1\phi,
\end{align}
where the relation \eqref{eq:key-one} has been used.
On the other hand, we have
\begin{equation}
 \widetilde{\phi[1]} =\widetilde{\widetilde\phi}+\widetilde r_1 \widetilde\phi.
\end{equation}
Combining them together yields
\begin{align}
 \widehat{\phi[1]} &=\widetilde{\phi[1]}
 + \left(\widehat r_1 - \widetilde r_1 -\widehat{\widetilde u}_{0}+\widetilde{\widetilde u}_{0} \right)\widetilde\phi
 -\bigl(\widehat{u[1]}-\widetilde{u[1]}\bigr)r_1\phi \nonumber\\
 & =\widetilde{\phi[1]}
 - \bigl(\widehat{u[1]}-\widetilde{u[1]} \bigr)\widetilde\phi
 -\bigl(\widehat{u[1]}-\widetilde{u[1]}\bigr)r_1\phi \nonumber\\
 & =\widetilde{\phi[1]}-\bigl(\widehat{u[1]}-\widetilde{u[1]}\bigr)\phi[1].
\end{align}
Hence,
the pair $(u[1],\phi[1])$ satisfies the same Lax triplet \eqref{eq:laxpair-H1}.

Once we get the one-step Darboux transformation,
it is then routine work to iterate it successively and get results
for the $N$-fold Darboux transformation.
By the discrete Crum theorem (cf.\cite{ZPZ-TMP-2020}), the resulting $N$-fold transformation
admits a compact representation in terms of Casoratians.
We introduce the eigenfunction vector
\begin{equation}\label{eq:phi-e-v}
 \boldsymbol{\phi}
 =\bigl(\phi_{1},\phi_{2},\ldots,\phi_{N},\phi\bigr)^{\mathrm T}
\end{equation}
and a $N+1$ component  constant vector
\begin{equation}\label{sigma}
 \boldsymbol{\sigma}=(0,\ldots,0,1)^{\mathrm T}.
\end{equation}
Define the Casoratian determinants
\begin{equation}\label{eq:casoratian-F}
\begin{split}
 \left|
 \widehat{N-1}(\boldsymbol{\phi}),
 \boldsymbol{\sigma}\right|_{[a]}
 &:=\left|\boldsymbol{\phi},~
 E_{[a]}\boldsymbol{\phi},~
 E_{[a]}^2\boldsymbol{\phi},~
 \ldots,
 E_{[a]}^{N-1}\boldsymbol{\phi},~
 \boldsymbol{\sigma}
 \right|,
\end{split}
\end{equation}
and
\begin{equation}\label{eq:casoratian-H}
\begin{split}
 \left| \widehat{N}(\boldsymbol{\phi}) \right|_{[a]}
 &:=\left|
 \boldsymbol{\phi},~
 E_{[a]}\boldsymbol{\phi},~
 E_{[a]}^2\boldsymbol{\phi},~
 \ldots,
 E_{[a]}^N\boldsymbol{\phi}
 \right|,
\end{split}
\end{equation}
where $a\in\{n,m\}$.
With these shorthand notations,
the results for the $N$-fold Darboux transformation can be presented in the following theorem
(cf.\cite{ZPZ-TMP-2020}).

\begin{theorem}\label{thm:H1}
Let $\phi_{1},\ldots,\phi_{N}$ be the eigenfunctions of \eqref{eq:laxpair-H1}
corresponding to $u=u_0$ and the distinct spectral parameters
$\mathbf{k}_1^2,\ldots,\mathbf{k}_N^2$. Then the $N$-fold Darboux
transformation gives
\begin{equation}
\begin{split}
 \phi[N]
 &=
 \frac{
 \left|\widehat{N}(\boldsymbol{\phi})\right|_{[n]}}
 {\left| \widehat{N-1}(\boldsymbol{\phi}),~ \boldsymbol{\sigma}\right|_{[n]} },\\
 u[N]&=E_{[n]}^N u_{0}
+ \frac{ \left|
\widehat{N-2}(\boldsymbol{\phi}),
 N(\boldsymbol{\phi}),
 \boldsymbol{\sigma}
 \right|_{[n]}}
 {\left|
 \widehat{N-1}(\boldsymbol{\phi}),
 \boldsymbol{\sigma}\right|_{[n]}}.
\end{split}
\end{equation}
The transformed pair $(u[N],\phi[N])$ satisfies the Lax triplet \eqref{eq:laxpair-H1}.
 Consequently, $u[N]$ provides a new solution of the lpKdV equation \eqref{eq:H1}.
\end{theorem}

\subsection{Elliptic solitons}\label{sec:2-3}

We now specialize the preceding Darboux--Crum construction to an elliptic background solution $u_0$.

Consider the Weierstrass functions $\sigma(z)$, $\zeta(z)$ and $\wp(z)$,
associated with a fixed $(2w, 2w')$-period lattice,
where $w$ and $w'$ are the two half periods of the elliptic function $\wp(z)$.
The identities for these functions used below are collected in Appendix \ref{app:det-elliptic}
of this paper.
The spectral parameter $\mathrm{\mathbf{k}}$ and the lattice parameters $\alpha$ and $\beta$ of the lpKdV equation are parameterized by the elliptic variables $k$, $a$ and $b$, respectively, through
\begin{equation}\label{eq:ell-para}
\mathrm{\mathbf{k}}^2=\wp(k)-e, \qquad \alpha^2=\wp(a)-e, \qquad \beta^2=\wp(b)-e.
\end{equation}
When $e$ is chosen as the value of $\wp(z)$ at a half-period $w$, i.e., $e=\wp(w)$,
 up to consistent choices of signs, the relation \eqref{eq:Lame-half-period} gives\cite{NA-IMRN-2010}
\begin{equation}
\mathrm{\mathbf{k}}=\Psi_k(w),\qquad
\alpha=\Psi_a(w),\qquad
\beta=\Psi_b(w).
\end{equation}
The lpKdV equation \eqref{eq:H1} admits the
background solution \cite{NA-IMRN-2010}
\begin{equation}\label{eq:u0e}
u_{0,[e]}=\zeta(\xi)-n\zeta(a)-m\zeta(b),
\qquad
\xi=na+mb+\xi_0,
\end{equation}
where $\xi_0\in \mathbb{C}$ and  the subscript $[e]$ stands for the elliptic case\footnote{
Note that the lpKdV equation
yields the potential KdV equation in continuum limit, and strictly speaking,
\eqref{eq:u0e} is not an elliptic function. Nevertheless, for fixed $n,m\in\mathbb Z$ and $\xi_0$, $\widetilde{\widetilde u}_{0,[e]}-u_{0,[e]}$ and $\widehat{\widehat u}_{0,[e]}-u_{0,[e]}$ are elliptic in each of the lattice parameters $a$ and $b$.}.
The  Lax triplet \eqref{eq:laxpair-H1} admits a discrete Lam\'e-type eigenfunction \cite{LZ-2024}
\begin{equation}\label{eq:phi-e}
\phi_{[e]}(k):=
\rho_0^{+}\left(\Phi_{a}(k)\right)^n
\left(\Phi_{b}(k)\right)^m
\Phi_{\xi}(-k)
+\rho_0^{-}\left(\Phi_{a}(-k)\right)^n
\left(\Phi_{b}(-k)\right)^m
\Phi_{\xi}(k),
\end{equation}
where  $\rho_0^+$ and $\rho_0^-$  are arbitrary constants.
Here, $k$ is assumed not to be any half-period  modulo the period lattice, ensuring the linear independence of the two eigenfunction components.
We omit the parameter $k$ when no confusion arises.
Choose distinct spectral parameters $\mathbf{k}_i^2=\wp(k_i)-e$,
and let
\begin{equation}\label{2.34}
\phi_{i,[e]}=\rho_{i,0}^{+}\left(\Phi_{a}(k_i)\right)^n
\left(\Phi_{b}(k_i)\right)^m
\Phi_{\xi}(-k_i)
+\rho_{i,0}^{-}\left(\Phi_{a}(-k_i)\right)^n
\left(\Phi_{b}(-k_i)\right)^m
\Phi_{\xi}(k_i),
\end{equation}
be the corresponding eigenfunctions.
The eigenfunction vector is defined as
 \begin{equation}\label{eq:seed-phi-ki}
\boldsymbol{\phi_{[e]}}=(\phi_{1,[e]},~\phi_{2,[e]},~\ldots,~\phi_{N,[e]},~\phi_{[e]})^T.
\end{equation}
Applying Theorem~\ref{thm:H1} to the above choices for $u_0$ and eigenfunctions
gives rise to the pair $(\phi[N],u[N])$ in terms of Casoratians,
which are new solutions of the Lax triplet \eqref{eq:laxpair-H1}.
In particular, $u[N]$ is an elliptic $N$-soliton solution of the lpKdV equation \eqref{eq:H1}
(cf.\cite{NA-IMRN-2010,LZ-2024}).

\begin{theorem}\label{thm:ellsol-H1}
The Lax triplet \eqref{eq:laxpair-H1} admits the following solutions in Casoratians:
\begin{align}\label{psi-u-ENDT-H1}
\phi[N]=\frac{|\widehat{N}(\boldsymbol{\phi}_{[e]})|_{[n]}}{|\widehat{N-1}(\boldsymbol{\phi}_{[e]}),
~\boldsymbol{\sigma}|_{[n]}}, \qquad
u[N]=E_{[n]}^Nu_{0,[e]}
+\frac{|\widehat{N-2}(\boldsymbol{\phi}_{[e]}),~N(\boldsymbol{\phi}_{[e]}),~\boldsymbol{\sigma}|_{[n]}}
{|\widehat{N-1}(\boldsymbol{\phi}_{[e]}),~\boldsymbol{\sigma}|_{[n]}}.
\end{align}
\end{theorem}

\section{Elliptic soliton solutions of the leKdV equations}\label{sec:3}

The Casoratian eigenfunctions obtained in the preceding section can be further used to construct solutions of a class of related  lattice equations.
As shown in \cite{WZZZ-JPA-2024},
the lpKdV Lax triplet \eqref{eq:laxpair-H1} leads to two classes of KdV-type
lattice eigenfunction equations and their $\delta$-extensions, namely,
leKdV-I, leKdV-I($\delta$), leKdV-II, and leKdV-II($\delta$).
These equations are associated, respectively, with the ABS equations H3(0), H3($\delta$), Q3(0), and Q3($\Delta$).
In this section, we use the elliptic eigenfunctions \eqref{eq:phi-e} and the Casoratian formula in
\eqref{psi-u-ENDT-H1} to construct explicit elliptic soliton solutions for these four equations.

\subsection{leKdV-I and H3(0)}\label{sec:3-1}

We first consider the leKdV-I equation, which is
an alternative form of the lattice mKdV (lmKdV) equation \cite{HJN-book-2016}
or H3(0) equation in the ABS list \cite{ABS-CMP-2003}.

Let us briefly recall the construction of the leKdV-I equation \cite{WZZZ-JPA-2024}.
Consider the Lax triplet \eqref{eq:laxpair-H1}
with parametrisation  \eqref{eq:ell-para},
i.e.
\begin{subequations}\label{eq:laxpair-H1-ell}
\begin{align}
&\widetilde{\widetilde{\phi}}
+(\widetilde{\widetilde{u}}-u)\widetilde{\phi}
+\wp(a)\phi=\wp(k) \phi, \label{eq:laxpair-H1-A-e}\\
&\widehat{\widehat{\phi}}
+(\widehat{\widehat{u}}-u)\widehat{\phi}
+\wp(b)\phi=\wp(k) \phi, \label{eq:laxpair-H1-B-e}\\
&\widehat{\phi}=\widetilde{\phi}-(\widehat{u}-\widetilde{u})\phi. \label{eq:laxpair-H1-C-e}
\end{align}
\end{subequations}
Let $\varphi_1$ and $\varphi_2$ be two linearly independent solutions of the above triplet.
Introduce
\begin{equation}
[ f, ~g]_{[j_1],[j_2]}^{\pm} :=
( E_{[j_1]} f )( E_{[j_2]} g )
\pm ( E_{[j_2]} f )( E_{[j_1]} g ), \qquad j_1, j_2 \in \{0, n, m\},
\end{equation}
where $E_{[0]} := \mathrm{id}$ represents the identity operator.
Applying $E_{[n]}$ to $[ \varphi_1,~\varphi_2]_{[0],[n]}^-$ and substituting
$\widetilde{\widetilde{\varphi}}_i$ from \eqref{eq:laxpair-H1-A-e}, the
term proportional to $\widetilde{\widetilde{u}}-u$ cancels identically,
leaving
\begin{equation*}
  E_{[n]}[ \varphi_1, ~ \varphi_2]_{[0],[n]}^- = (\wp(a)-\wp(k))\,[ \varphi_1, ~\varphi_2]_{[0],[n]}^-.
\end{equation*}
The same computation applied to $[ \varphi_1,\varphi_2]_{[0],[m]}^-$ and using
\eqref{eq:laxpair-H1-B-e}  yields
\begin{equation*}
  E_{[m]}[ \varphi_1,~ \varphi_2]_{[0],[m]}^- = (\wp(b)-\wp(k))\,[ \varphi_1, ~ \varphi_2]_{[0],[m]}^-.
\end{equation*}
The two relations indicate the discrete Abel formula (cf.\cite{WZZZ-JPA-2024})
\begin{equation}\label{eq:Abel-1}
[ \varphi_1,~ \varphi_2]_{[0],[n]}^-=[ \varphi_1,~ \varphi_2]_{[0],[m]}^-
=\rho(k)(\wp(a)-\wp(k))^n(\wp(b)-\wp(k))^m,
\end{equation}
where $\rho(k)$ is an `integration' constant independent of the discrete variables $ n$ and $m$.
Note that the equality of the two Casoratians in \eqref{eq:Abel-1} follows directly from the third equation \eqref{eq:laxpair-H1-C-e} of the  Lax triplet.

The discrete Abel formula \eqref{eq:Abel-1} provides a B\"acklund transformation
between the equations satisfied by $\varphi_1$ and $\varphi_2$.
Eliminating $\varphi_2$  from \eqref{eq:Abel-1}
yields the  equation for $\varphi_1$, and vice versa for $\varphi_2$.
This leads to the same equation, namely, the leKdV-I equation \cite{WZZZ-JPA-2024}
 \begin{equation}\label{eq:leKdV-I}
(\wp(a)-\wp(k)) \varphi\widehat{\varphi} -(\wp(b)-\wp(k)) \varphi \widetilde{\varphi} - \widetilde{\varphi} \widehat{\widetilde{\varphi}} +\widehat{\varphi} \widehat{\widetilde{\varphi}} = 0,
 \end{equation}
where we have used $\varphi$ to replace $\varphi_i$, $i=1,2$.
Furthermore, the leKdV-I equation \eqref{eq:leKdV-I} is converted into the lmKdV/H3(0) equation,
\begin{equation}\label{eq:H30}
s(v\widehat{v} + \widetilde{v}\widehat{\widetilde{v}}) - t(v\widetilde{v} + \widehat{v}\widehat{\widetilde{v}}) = 0,
\end{equation}
under the transformation
 \begin{equation}
\varphi = (\mathrm{i}s)^n (\mathrm{i}t)^m v,
 \end{equation}
where $\mathrm{i}^2=-1$, and
\begin{equation}\label{eq:st}
 s^2=\wp(a)-\wp(k)=\alpha^2-\mathbf{k}^2, \qquad  t^2=\wp(b)-\wp(k)=\beta^2-\mathbf{k}^2.
 \end{equation}
Here and hereafter we take the same square root branch for $s$
and same square root branch for $t$.

Obviously, using the results of Theorem \ref{thm:ellsol-H1}, one can write out
explicit forms for $\varphi_1$ and $\varphi_2$ when $u=u_{0,[e]}$ takes \eqref{eq:u0e}.
To be more specific, we introduce
\begin{subequations}\label{eq:seed-wave}
\begin{align}
\varphi_{1,[e]}(k)&=\left(\Phi_{a}(k)\right)^n
\left(\Phi_{b}(k)\right)^m
\Phi_{\xi}(-k),\\
\varphi_{2,[e]}(k)&=\left(\Phi_{a}(-k)\right)^n
\left(\Phi_{b}(-k)\right)^m
\Phi_{\xi}(k),
\end{align}
\end{subequations}
such that the elementary eigenfunction \eqref{eq:phi-e} is expressed as
\begin{equation*}
\phi_{[e]} = \rho_0^+\varphi_{1,[e]} +\rho_0^- \varphi_{2,[e]}.
\end{equation*}
The two functions in \eqref{eq:seed-wave} are functionally independent
when $\wp'(k)\neq 0$.
In fact, substituting them into the discrete Abel formula \eqref{eq:Abel-1}
yields $\rho(k)= -\wp'(k)$, which is not zero when $k$ is neither the
half-period $w$ nor $w'$, nor $w+w'$.
In the above calculation, the formulae \eqref{eq:add-5} and \eqref{eq:add-1}
and the following one,
\begin{equation*}
\Phi_{\xi}(-k)\Phi_{a}(-k)\Phi_{\xi+a}(k)-\Phi_{\xi}(k)\Phi_{a}(k)\Phi_{\xi+a}(-k)
=\frac{\sigma(2k)}{\sigma(k)^4},
\end{equation*}
are used.

For the elliptic $N$-soliton solution to the leKdV-I equation  \eqref{eq:leKdV-I} and H3(0) equation  \eqref{eq:H30},
in light of Theorem \eqref{thm:ellsol-H1},
we have the following  result.

\begin{theorem}\label{thm:leKdV-I-N}
Define
\begin{equation}\label{eq:sol-leKdV-N}
  \varphi_{j}[N](k)
  := \frac{|\widehat{N}(\boldsymbol{\varphi}_j)|_{[n]}}{|\widehat{N-1}(\boldsymbol{\varphi}_j),
  ~\boldsymbol{\sigma}|_{[n]}},
  \qquad j = 1, 2,
\end{equation}
where the composing vectors are
\begin{equation}\label{phi-ell}
\boldsymbol{\varphi}_j
=(\phi_{1,[e]}, ~\phi_{2,[e]},~\ldots,~\phi_{N,[e]},~\varphi_{j,[e]}(k))^T, \qquad j=1,2,
\end{equation}
with $\phi_{i,[e]}$ given in \eqref{2.34} for $i=1,\cdots N$,
and $\boldsymbol{\sigma}$ given in \eqref{sigma}.
Both functions $\varphi_1[N](k)$ and $\varphi_2[N](k)$  are  elliptic $N$-soliton solutions
of the leKdV-I equation \eqref{eq:leKdV-I}.
In addition,
  \begin{equation}\label{eq:sol-H30-N}
    v_j[N](k) = (\mathrm{i}s)^{-n}(\mathrm{i}t)^{-m}\,\varphi_j[N](k),
    \qquad j = 1, 2,
  \end{equation}
  provide the elliptic $N$-soliton solutions for the H3(0) equation \eqref{eq:H30}.
\end{theorem}

\subsection{Determining \texorpdfstring{$\rho(k)$}{rho(k)} in the discrete Abel formula}\label{sec:3-2}

The `integration' constant $\rho(k)$ in the discrete Abel formula \eqref{eq:Abel-1}
plays a key role in the formulation of the parameter $\delta$ in the ABS equations.
In the following we determine the value of $\rho(k)$.

\begin{theorem}\label{thm:leKdV-I-rho}
The eigenfunctions  $\varphi_1[N](k)$ and $\varphi_2[N](k)$ defined in
Theorem \ref{thm:leKdV-I-N}
satisfy the discrete Abel formula
  \begin{equation}\label{eq:Abel-2}
    [\varphi_1[N](k),~\varphi_2[N](k)]_{[0],[n]}^-
    = [\varphi_1[N](k),~\varphi_2[N](k)]_{[0],[m]}^-
    = \rho_N(k)\,s^{2n}t^{2m},
  \end{equation}
  with the `integration' constant
    \begin{equation}\label{eq:rho-N}
    \rho_N(k) = -\wp'(k)\prod_{j=1}^{N}(\wp(k_j)-\wp(k))\\
\end{equation}
\end{theorem}

\begin{proof}
For a general positive integer $N$, we introduce an auxiliary $(N+2)\times(N+2)$ determinant
\begin{equation*}
G =
  \left|
\begin{matrix}
\phi_{1,[e]} & E_{[n]}\phi_{1,[e]}& \cdots & E_{[n]}^{N} \phi_{1,[e]} & E_{[n]}^{N+1}\phi_{1,[e]} \\
\vdots & \vdots & \ddots & \vdots & \vdots \\
\phi_{N,[e]} & E_{[n]}\phi_{N,[e]} & \cdots & E_{[n]}^{N} \phi_{N,[e]} &  E_{[n]}^{N+1}\phi_{N,[e]} \\
\varphi_{1,[e]} &E_{[n]} \varphi_{1,[e]}& \cdots &  E_{[n]}^{N} \varphi_{1,[e]}  & E_{[n]}^{N+1}\varphi_{1,[e]} \\
\varphi_{2,[e]}& E_{[n]}\varphi_{2,[e]}& \cdots & E_{[n]}^{N} \varphi_{2,[e]} &  E_{[n]}^{N+1}\varphi_{2,[e]}
\end{matrix}
\right|.
\end{equation*}
A direct calculation using the Desnanot--Jacobi identity \eqref{eq:DJ} yields
 \begin{equation*}
 \begin{split}
    [\varphi_1[N](k),~\varphi_2[N](k)]_{[0],[n]}^-
   &=\frac{G_{N+2}^{N+2}G_{N+1}^1-G_{N+2}^1G_{N+1}^{N+2}}
   {|\widehat{N-1}(\boldsymbol{\varphi}_j),~\boldsymbol{\sigma}|_{[n]}\cdot E_{[n]}|\widehat{N-1}(\boldsymbol{\varphi}_j),~\boldsymbol{\sigma}|_{[n]}}\\
   &=\frac{G}
   {|\widehat{N-1}(\boldsymbol{\varphi}_j),~\boldsymbol{\sigma}|_{[n]}}.
   \end{split}
 \end{equation*}
 It remains to evaluate $G$ explicitly.

 In \cite{LZ-2024}, we proved a determinant identity
  \begin{equation}\label{eq:det-ident}
\begin{split}
&\left|\Phi_{\xi}(-\mathbf{k}) ,~ \Phi_{\xi+a}(-\mathbf{k})\Phi_{a}(\mathbf{k}),~
\cdots ,\Phi_{\xi+(N-1)a}(-\mathbf{k})\Big(\Phi_{a}(\mathbf{k})\Big)^{N-1}  \right| \\
=\,&(-1)^N\frac{\sigma(\xi-\sum_{i=1}^Nk_i)}{\sigma(\xi)}\cdot
\frac{\prod_{1\leq i<j\leq N}\sigma(k_i-k_j)}{\sigma^{N}(k_1)\cdots\sigma^{N}(k_N)},
\end{split}
\end{equation}
where $f(\mathbf{k})$ denotes a column vector $(f(k_1), f(k_2), \cdots, f(k_N))^T$.
Applying \eqref{eq:det-ident} with $N $ replaced by $ N+2$ and the ordered parameters $(\epsilon_1 k_1,\cdots\epsilon_N k_N, k, -k)$,
we  immediately obtain
\begin{equation*}
\begin{split}
 & \left|
\begin{matrix}
\Phi_{\xi}(-\epsilon_1 k_1)& \Phi_{\xi+a}(-\epsilon_1 k_1)\Phi_{a}(\epsilon_1 k_1)& \cdots & \Phi_{\xi+(N+1)a}(-\epsilon_1 k_1)\left(\Phi_{a}(\epsilon_1 k_1)\right)^{N+1} \\
\vdots & \vdots & \ddots & \vdots \\
\Phi_{\xi}(-\epsilon_N k_N)& \Phi_{\xi+a}(-\epsilon_N k_N)\Phi_{a}(\epsilon_N k_N)& \cdots & \Phi_{\xi+(N+1)a}(-\epsilon_N k_N)\left(\Phi_{a}(\epsilon_N k_N)\right)^{N+1}  \\
\Phi_{\xi}(- k) & \Phi_{\xi+a}(- k)\Phi_{a}( k)& \cdots & \Phi_{\xi+(N+1)a}(- k)\left(\Phi_{a}( k)\right)^{N+1}  \\
\Phi_{\xi}(k)& \Phi_{\xi+a}(k)\Phi_{a}( -k)& \cdots & \Phi_{\xi+(N+1)a}( k)\left(\Phi_{a}(-k)\right)^{N+1}  \\
\end{matrix}
\right|\\
=&V_{[e]}\left(\epsilon_1 k_1,\cdots \epsilon_N k_N\right)\cdot
\prod_{i=1}^N\frac{\sigma(\epsilon_i k_i-k)\sigma(\epsilon_i k_i+k)}{\sigma^2(\epsilon_i k_i)\sigma(k)\sigma(-k)}\cdot
\frac{\sigma(2k)}{\sigma^4(k)}\\
=&V_{[e]}\left(\epsilon_1 k_1,\cdots \epsilon_N k_N\right)\cdot
\prod_{i=1}^N\left(\wp(k_i)-\wp(k)\right)\cdot
(-\wp'(k)),
\end{split}
\end{equation*}
where
\begin{equation*}
V_{[e]}\left(\epsilon_1 k_1,\cdots \epsilon_N k_N\right)=(-1)^N\frac{\sigma(\xi-\sum_{i=1}^N \epsilon_i k_i)}{\sigma(\xi)}\cdot
\frac{\prod_{1\leq i<j\leq N}\sigma(\epsilon_i k_i- \epsilon_j k_j)}{\sigma^{N}(\epsilon_1 k_1)\cdots\sigma^{N}(\epsilon_N k_N)}.
\end{equation*}
Since each $\phi_{i,[e]}$ is a sum of two elementary functions, the determinants $|\widehat{N-1}(\boldsymbol{\varphi}_j),~\boldsymbol{\sigma}|_{[n]}$ and  $G$  can be expanded into $2^N$ Casoratians.
Define
\begin{equation*}
  k_{i,+1}=\rho_{i,0}^{+}\left(\Phi_a(k_i)\right)^n\left(\Phi_b(k_i)\right)^m, \qquad
   k_{i,-1}=\rho_{i,0}^{-}\left(\Phi_a(-k_i)\right)^n\left(\Phi_b(-k_i)\right)^m.
\end{equation*}
Then we have
\begin{align}
&|\widehat{N-1}(\boldsymbol{\varphi}_j),~\boldsymbol{\sigma}|_{[n]}=\sum_{\boldsymbol{\epsilon}\in\{1,-1\}^{N}}
\left(\prod_{i=1}^N k_{i,\epsilon_i}\right)
V_{[e]}\left(\epsilon_1 k_1,\cdots \epsilon_N k_N\right),\\
&G=- \wp'(k) \prod_{i=1}^N (\wp(k_i)-\wp(k)) s^{2n}t^{2m}\left(\sum_{\boldsymbol{\epsilon}\in\{1,-1\}^{N}}
\left(\prod_{i=1}^N k_{i,\epsilon_i}\right)
V_{[e]}\left(\epsilon_1 k_1,\cdots \epsilon_N k_N\right)\right),
\end{align}
where $\boldsymbol{\epsilon}=(\epsilon_1,\cdots,\epsilon_N)$.
 Substituting these identities into the discrete Abel formula,
  the denominator cancels exactly, and we obtain
\begin{equation*}
  [\varphi_1[N](k),~\varphi_2[N](k)]_{[0],[n]}^-
   =\rho_N(k)s^{2n}t^{2m},
\end{equation*}
where
\begin{equation*}
 \rho_N(k)=- \wp'(k) \prod_{i=1}^N (\wp(k_i)-\wp(k)).
\end{equation*}
Finally, the third equation of the transformed Lax triplet, namely, \eqref{eq:laxpair-H1-C-e} with $u$ replaced
by $u[N]$ yields
\[
[\varphi_1[N](k),\varphi_2[N](k)]^-_{[0],[m]}
=
[\varphi_1[N](k),\varphi_2[N](k)]^-_{[0],[n]}
=
\rho_N(k)s^{2n}t^{2m},
\]
which establishes the $m$-direction equality in \eqref{eq:Abel-2}.
This completes the proof.
\end{proof}

\subsection{leKdV-I(\texorpdfstring{$\delta$}{delta}) and H3(\texorpdfstring{$\delta$}{delta})}
\label{sec:3-3}

We next consider the $\delta$-extension, namely leKdV-I($\delta$) equation \cite{WZZZ-JPA-2024},
\begin{equation}\label{eq:leKdV-Id}
s^{2}\psi\widehat{\psi}-t^{2}\psi\widetilde{\psi}
-\widetilde{\psi}\widehat{\widetilde{\psi}}+\widehat{\psi}\widehat{\widetilde{\psi}}
=(-1)^{n+m}\delta\left(s^{2}-t^{2}\right)s^{2n}t^{2m},
\end{equation}
where $s$ and $t$ are defined by \eqref{eq:st}. Under the transformation
\begin{equation}\label{eq:IITOH3}
\psi=(\mathrm{i}s)^n(\mathrm{i}t)^m v,
\end{equation}
the  leKdV-I($\delta$) equation \eqref{eq:leKdV-Id} is transformed into
\begin{equation}\label{eq:H3d}
s(v\widehat{v}+\widetilde{v}\widehat{\widetilde{v}})
-t(v\widetilde{v}+\widehat{v}\widehat{\widetilde{v}})
=-\mathrm{i}\delta\left(\frac{s}{t}-\frac{t}{s}\right),
\end{equation}
which coincides with the H3($\delta$) equation in the ABS classification.

Let us recall how equation \eqref{eq:leKdV-Id} is constructed (cf. \cite{WZZZ-JPA-2024}).
In addition to the Lax triplet \eqref{eq:laxpair-H1-ell}, we introduce its `adjoint' form
\begin{subequations}\label{eq:laxpair-H1-ella}
\begin{align}
&\widetilde{\widetilde{\phi}}{}^*
-(\widetilde{\widetilde{u}}-u)\widetilde{\phi}^*
+\wp(a)\phi^* =\wp(k) \phi^*, \label{eq:laxpair-H1-A-ea}\\
&\widehat{\widehat{\phi}}{}^*
-(\widehat{\widehat{u}}-u)\widehat{\phi}^*
+\wp(b)\phi^* =\wp(k) \phi^*, \label{eq:laxpair-H1-B-ea}\\
&\widehat{\phi}^*=\widetilde{\phi}^*+(\widehat{u}-\widetilde{u})\phi^*. \label{eq:laxpair-H1-C-ea}
\end{align}
\end{subequations}
where $\phi^*$ stands for an `adjoint' counterpart  of $\phi$.
Apparently, for any solution $\phi$ of \eqref{eq:laxpair-H1-ell},
$\phi^*=(-1)^{n+m}\phi$ is a solution of \eqref{eq:laxpair-H1-ella}, and vice versa.
In other words, for given $u$ and $k$, for any independent solutions of \eqref{eq:laxpair-H1-ell}, namely,
$\varphi_1$ and $\varphi_2$, then,
\begin{equation}\label{phi*}
\phi^{*}=(-1)^{n+m}(C\varphi_{1}+D\varphi_{2}),~~ C,D\in \mathbb{C}
\end{equation}
provides a general solution to \eqref{eq:laxpair-H1-ella}.
On the one hand, similar to \eqref{eq:Abel-1},
for any $\phi$ and $\phi^*$ satisfying \eqref{eq:laxpair-H1-ell} and \eqref{eq:laxpair-H1-ella} with the same
$u$ and $k$, one can prove that
\begin{equation}\label{eq:Abel-3}
[\phi, ~\phi^{*}]_{[0],[n]}^+
=[\phi,~\phi^{*}]_{[0],[m]}^+
=(-1)^{n+m}	\varrho(k) s^{2n}t^{2m},
\end{equation}
where $\varrho(k)$ is the constant due to `integration'.
On the other hand, for any independent solutions $\varphi_1$ and $\varphi_2$
of \eqref{eq:laxpair-H1-ell}, considering
\begin{equation}\label{phi}
\phi=A \varphi_{1}+B\varphi_{2},~~ A,B\in \mathbb{C}
\end{equation}
together with \eqref{phi*}, one can prove that
\begin{equation}\label{eq:Abel-4}
[\phi,~\phi^{*}]_{[0],[n]}^+
=-(-1)^{n+m}(AD-BC)[\varphi_1,~\varphi_2]_{[0],[n]}^-,
\end{equation}
where we have assumed $AD-BC\neq 0$.
Recalling the Abel formula \eqref{eq:Abel-1} we immediately find
\[ \varrho(k)=-\rho(k)(AD-BC).\]

Next we introduce functions
\begin{subequations}
\begin{align}
&\psi=\phi+\phi^*=\left(A+(-1)^{n+m}C\right)\varphi_{1}
+\left(B+(-1)^{n+m}D\right)\varphi_{2},\\
&\psi^{*}=\phi-\phi^*=\left(A-(-1)^{n+m}C\right)\varphi_{1}
+\left(B-(-1)^{n+m}D\right)\varphi_{2},
\end{align}
\end{subequations}
where $\phi$ and $\phi^*$ are defined in \eqref{phi} and \eqref{phi*}, as we assume $AD-BC\neq 0$.
Then, using the discrete Abel formula \eqref{eq:Abel-3}, we obtain
\begin{equation}\label{3.27}
\psi\widetilde{\psi} - \psi^* \widetilde{\psi}^* = \psi \widehat{\psi} - \psi^* \widehat{\psi}^*
= \delta (-1)^{n+m} s^{2n} t^{2m},
\end{equation}
where
 \begin{equation}\label{eq:I-delta}
\delta
=2\varrho(k).
\end{equation}
The above equation \eqref{3.27} serves as a pair of B\"acklund transformations.
Eliminating $\psi^{*}$ shows that $\psi$ satisfies the leKdV-I($\delta$) equation \eqref{eq:leKdV-Id};
likewise, $\psi^{*}$ satisfies the leKdV-I($-\delta$) equation.

In the simplest case, for the elementary eigenfunctions $\varphi_{1,[e]}$ and $\varphi_{2,[e]}$
defined in \eqref{eq:seed-wave}
we introduce
\begin{equation}
\phi_{[e]}=A\varphi_{1,[e]}+B\varphi_{2,[e]},\qquad
\phi_{[e]}^{*}=(-1)^{n+m}(C\varphi_{1,[e]}+D\varphi_{2,[e]}),
\end{equation}
which are the solutions to the triplets \eqref{eq:laxpair-H1-ell} and  \eqref{eq:laxpair-H1-ella},
respectively, where $u=u_{0,[e]}$ taken as in \eqref{eq:u0e}.
Then,
\begin{align}
\psi_{[e]}&=\phi_{[e]}+\phi^*_{[e]}=\left(A+(-1)^{n+m}C\right)\varphi_{1,[e]}(k)
+\left(B+(-1)^{n+m}D\right)\varphi_{2,[e]}(k),\\
\psi_{[e]}^{*}&=\phi_{[e]}-\phi^*_{[e]}=\left(A-(-1)^{n+m}C\right)\varphi_{1,[e]}(k)
+\left(B-(-1)^{n+m}D\right)\varphi_{2,[e]}(k),
\end{align}
respectively provide plane wave solutions for the leKdV-I($\delta$) equation \eqref{eq:leKdV-Id}
and the leKdV-I(-$\delta$) equation, where
\begin{equation}\label{eq:I-delta-e}
\delta
=2\wp'(k)(AD-BC).
\end{equation}

Now, taking $\varphi_{1}[N](k)$ and $\varphi_{2}[N](k)$ to be the two solutions
in \eqref{eq:sol-leKdV-N} in Theorem \ref{thm:leKdV-I-N},
and making use of the result in Sec. \ref{sec:3-2},
it is easy to arrive at the Casoratian expression for the elliptic solitons
of the leKdV-I($\delta$) equation as well as the H3($\delta$) equation.

\begin{theorem}\label{thm:ellsol-H3d}
The elliptic $N$-soliton solution of the leKdV-I($\delta$) equation \eqref{eq:leKdV-Id} can be given as
\begin{equation}\label{eq:leKdV-I-psi}
\psi[N]
=
\frac{\left(A+(-1)^{n+m}C\right)\varphi_{1}[N](k)
+
\left(B+(-1)^{n+m}D\right)\varphi_{2}[N](k)}{\mathcal N_N(k)},
\end{equation}
where $\varphi_{j}[N](k)$, $j=1,2$ are given by \eqref{eq:sol-leKdV-N},
\begin{equation}\label{eq:factor}
\mathcal N_N(k)=\prod_{j=1}^{N}(\wp(k_j)-\wp(k))^{\frac12}
\end{equation}
is a normalization factor chosen so that $\delta$ can be expressed as \eqref{eq:I-delta-e}.
Under the  transformation \eqref{eq:IITOH3}, the elliptic $N$-soliton solution of the H3($\delta$) equation \eqref{eq:H3d} is
\begin{equation}
v[N]=(\mathrm{i}s)^{-n}(\mathrm{i}t)^{-m} \psi[N].
\end{equation}
\end{theorem}

\begin{remark}
Without the normalization factor $\mathcal N_N(k)$ in \eqref{eq:leKdV-I-psi}, the expression
\begin{equation*}
\Psi[N]
=\left(A+(-1)^{n+m}C\right)\varphi_{1}[N](k)
+ \left(B+(-1)^{n+m}D\right)\varphi_{2}[N](k)
\end{equation*}
satisfies the leKdV-I($\delta$) equation \eqref{eq:leKdV-Id} with $\delta$ replaced by
  \[
\delta_N=-2\rho_N(k)(AD-BC),
\]
where $\rho_N(k)$ is given in \eqref{eq:rho-N}.
\end{remark}

\subsection{leKdV-II(\texorpdfstring{$\delta$}{delta}) and Q3(\texorpdfstring{$\Delta$}{Delta})}
\label{sec:3-4}

We now turn to the second class of lattice eigenfunction equations,
the leKdV-II($\delta$) equation, which reads \cite{WZZZ-JPA-2024},
\begin{equation}\label{eq:II-delta}
s^{2}\mu^{2}Y\widehat{Y}
+\widetilde{Y}\widehat{\widetilde{Y}}
-t^{2}\nu^{2}Y\widetilde{Y}
-\widehat{Y}\widehat{\widetilde{Y}}
=
(\wp(a)-\wp(b))
\left(
Y\widehat{\widetilde{Y}}
+\widetilde{Y}\widehat{Y}
+\delta(\wp(k)-\wp(q))^2
s^{2n}\mu^{2n}
t^{2m}\nu^{2m}
\right),
\end{equation}
where $\delta$ is an arbitrary constant, $s^2$ and $t^2$ are given by \eqref{eq:st},
\begin{equation}\label{eq:munu}
\mu^{2}=\wp(a)-\wp(q)=\alpha^2-\mathbf{q}^2,
~~
\nu^{2}=\wp(b)-\wp(q)=\beta^2-\mathbf{q}^2,
~~
\mathrm{\mathbf{q}}^2=\wp(q)-e, ~~ q\in \mathbb{C}.
\end{equation}
The  case $\delta=0$ is referred to as the leKdV-II equation \cite{WZZZ-JPA-2024},
\begin{equation}\label{eq:leKdV-II}
s^{2}\mu^{2}y\widehat{y}
+\widetilde{y}\widehat{\widetilde{y}}
-t^{2}\nu^{2}y\widetilde{y}
-\widehat{y}\widehat{\widetilde{y}}
-(\wp(a)-\wp(b))
(\widetilde{y}\widehat{y}
+y\widehat{\widetilde{y}})
=0.
\end{equation}
We first recall how $Y$ and $y$ are formulated via eigenfunctions.
In addition to the triplet \eqref{eq:laxpair-H1-ell}, we consider a same one but replace
the spectral parameter $\mathbf{k}^2=\wp(k)-e$ by $\mathbf{q}^2=\wp(q)-e$,
i.e.
\begin{subequations}\label{eq:laxpair-H1-ellq}
\begin{align}
&\widetilde{\widetilde{\phi}}
+(\widetilde{\widetilde{u}}-u)\widetilde{\phi}
+\wp(a)\phi=\wp(q) \phi, \label{eq:laxpair-H1-A-eq}\\
&\widehat{\widehat{\phi}}
+(\widehat{\widehat{u}}-u)\widehat{\phi}
+\wp(b)\phi=\wp(q) \phi, \label{eq:laxpair-H1-B-eq}\\
&\widehat{\phi}=\widetilde{\phi}-(\widehat{u}-\widetilde{u})\phi. \label{eq:laxpair-H1-C-eq}
\end{align}
\end{subequations}
Same as the discrete Abel formula \eqref{eq:Abel-1}, here we have
\begin{equation}\label{eq:Abel-5}
[ \varphi_{1}(q), ~\varphi_{2}(q)]_{[0],[n]}^-=[ \varphi_{1}(q),~  \varphi_{2}(q)]_{[0],[m]}^-
=\rho(q)\mu^{2n}\nu^{2m},
\end{equation}
where $\varphi_{1}(q)$ and $\varphi_{2}(q)$ are two independent solutions of \eqref{eq:laxpair-H1-ellq}.
Introduce
\begin{align}
&\langle \boldsymbol{\varphi}(k)|=
(
\varphi_{1}(k),~\varphi_{2}(k)
),
\quad
|\boldsymbol{\varphi}(q)\rangle=
\begin{pmatrix}
\varphi_{1}(q)\\
\varphi_{2}(q)
\end{pmatrix},
\quad
M=
\begin{pmatrix}
A&B\\
C&D
\end{pmatrix},
\end{align}
where $A,B,C,D$ are arbitrary constants, and define
\begin{equation}\label{eq:Y}
Y=\langle\boldsymbol{\varphi}(k)|~M~|\widetilde{\boldsymbol{\varphi}}(q)\rangle
     -\langle\widetilde{\boldsymbol{\varphi}}(k)|~M~|\boldsymbol{\varphi}(q)\rangle.
\end{equation}
It has been shown in \cite{WZZZ-JPA-2024} that $Y$ solves the leKdV-II($\delta$) equation \eqref{eq:II-delta} with
  \begin{equation}\label{eq:II-deltavalue-g}
    \delta = \rho(k)\rho(q)(AD-BC).
  \end{equation}
In particular, $y=Y|_{|M|=0}$ yields a solution to the  leKdV-II equation \eqref{eq:leKdV-II}.

In the simplest case, taking
\begin{subequations}\label{Y}
\begin{equation}\label{eq:Y-Lame}
Y=\langle\boldsymbol{\varphi}_{[e]}(k)|~M~|\widetilde{\boldsymbol{\varphi}}_{[e]}(q)\rangle
     -\langle\widetilde{\boldsymbol{\varphi}}_{[e]}(k)|~M~|\boldsymbol{\varphi}_{[e]}(q)\rangle,
\end{equation}
where
 \begin{align}
&\langle \boldsymbol{\varphi}_{[e]}(k)|=
(
\varphi_{1,[e]}(k),~\varphi_{2,[e]}(k)
),
\quad
|\boldsymbol{\varphi}_{[e]}(q)\rangle=
\begin{pmatrix}
\varphi_{1,[e]}(q)\\
\varphi_{2,[e]}(q)
\end{pmatrix},
\end{align}
\end{subequations}
in which $\varphi_{1,[e]}(z)$ and $\varphi_{2,[e]}(z)$ are defined as in \eqref{eq:seed-wave},
we have the plane wave solution for the  leKdV-II($\delta$) equation \eqref{eq:II-delta} with
\begin{equation}\label{eq:II-deltavalue}
    \delta = \wp'(k)\wp'(q)(AD-BC).
\end{equation}
Note that the above $Y$ can be written as
\begin{equation}\label{eq:Yyij-e}
    Y=Ay_{11,[e]}+By_{12,[e]}+Cy_{21,[e]}+Dy_{22,[e]},
\end{equation}
where
\begin{equation}\label{eq:yij-e}
    y_{ij,[e]}=[\varphi_{i,[e]}(k),~\varphi_{j,[e]}(q)]^{-}_{[0],[n]},
    \qquad i,j\in\{1,2\}.
\end{equation}
For the  leKdV-II  equation \eqref{eq:leKdV-II}, its plane wave solution is
given by $y=Y|_{|M|=0}$.

In the general case, we can write \eqref{eq:Y} into the form
\begin{equation}\label{eq:Yyij}
    Y=Ay_{11}+By_{12}+Cy_{21}+Dy_{22},
\end{equation}
where
\begin{equation}\label{eq:yij}
    y_{ij}=[\varphi_{i}(k),~\varphi_{j}(q)]^{-}_{[0],[n]},
    \qquad i,j\in\{1,2\}.
\end{equation}
After replacing the eigenfunctions  $\varphi_{i}(k)$ and $\varphi_{j}(q)$ by their Darboux iterates $\varphi_i[N](k)$ and $\varphi_j[N](q)$, respectively, for $i,j=1,2$, we present  the
 elliptic $N$-soliton solutions in the following.

\begin{theorem}
Define
\begin{equation}\label{eq:yijN}
y_{ij}[N]=[\varphi_{i}[N](k),~\varphi_{j}[N](q)]_{[0],[n]}^-
\qquad
i,j\in\{1,2\},
\end{equation}
where $\varphi_{i}[N](z)$ is defined as in \eqref{eq:sol-leKdV-N}.
Then the leKdV-II($\delta$) equation
\eqref{eq:II-delta} admits the elliptic $N$-soliton solution
\begin{equation}\label{eq:YN}
Y[N]=
\frac{A y_{11}[N]+B y_{12}[N]+C y_{21}[N]+D y_{22}[N]}
{\mathcal N_N(k)\mathcal N_N(q)}
\end{equation}
where  $\mathcal N_N(z)$ is defined as in \eqref{eq:factor} and $\delta$ is given in \eqref{eq:II-deltavalue}.
In particular, the explicit solution of the leKdV-II equation \eqref{eq:leKdV-II} is recovered by setting $AD-BC=0$.
\end{theorem}

The leKdV-II($\delta$) equation \eqref{eq:II-delta} is related to the Q3($\Delta$) equation by a transformation introduced in \cite{WZZZ-JPA-2024}.
 The explicit elliptic solutions obtained above therefore give solutions of the Q3($\Delta$) equation.

\begin{theorem}\label{thm:Q3}
The Q3($\Delta$) equation reads \cite{HJN-book-2016}
\begin{equation}\label{eq:Q3delta}
  P(X\widehat{X}+\widetilde{X}\widehat{\widetilde{X}})
  -Q(X\widetilde{X}+\widehat{X}\widehat{\widetilde{X}}) =  (\alpha^2-\beta^2)\left(\widetilde{X}\widehat{X}+X\widehat{\widetilde{X}}+\frac{\Delta}{PQ} \right)
\end{equation}
where
$P=s\mu$, $ Q=t\nu$,
\begin{equation}\label{eq:Q3-deltavalue}
\Delta=(\wp(k)-\wp(q))^2\delta,
\end{equation}
$\alpha^2$ and $\beta^2$ are given by
\eqref{eq:ell-para}, and $\delta$ takes the form  \eqref{eq:II-deltavalue}.
With this parametrization,
$(\alpha,P),(\beta,Q)$  lie on the Jacobi quartic $\Gamma_{\mathbf k,\mathbf q}$
defined in \eqref{eq:Q3-ell-curve}.
The $Q3(\Delta)$ equation \eqref{eq:Q3delta} admits the following explicit solutions.

(i) The plane wave solution is given by
\begin{equation}\label{eq:X-Y}
X=s^{-n}t^{-m}\mu^{-n}\nu^{-m}Y,
\end{equation}
where the function $Y$ is given by \eqref{Y}.

(ii) The elliptic $N$-soliton solution in terms of Casoratian is
\begin{equation}
X[N]
=s^{-n}t^{-m}\mu^{-n}\nu^{-m}Y[N],
\end{equation}
where $Y[N]$ is defined by \eqref{eq:YN}.

In particular, when $AD-BC=0$, the above solutions reduce to the corresponding solutions of the Q3$(0)$ equation.
\end{theorem}

It is notable that, by applying the transformation
\begin{equation}
y=\left(\Phi_a(-k)\Phi_a(-q)\right)^n\left(\Phi_b(-k)\Phi_b(-q)\right)^m S
\end{equation}
 to the leKdV-II equation  \eqref{eq:leKdV-II} and using relation \eqref{eq:add-1}, we obtain the following equation
 \begin{equation}
 \left(\Phi_a(k) S-\Phi_a(-q)\widetilde{S}\right) \left(\Phi_a(q) \widehat{S}-\Phi_a(-k)\widetilde{\widehat{S}}\right)
 = \left(\Phi_b(k) S-\Phi_b(-q)\widehat{S}\right) \left(\Phi_b(q) \widetilde{S}-\Phi_b(-k)\widetilde{\widehat{S}}\right).
 \end{equation}
 This equation also appears in \cite{NA-IMRN-2010}, and is related to the
 Nijhoff--Quispel--Capel (NQC) equation introduced in \cite{NQC-PLA-1983}.

\section{Degeneration to the Q1(\texorpdfstring{$\Delta$}{Delta}) and Q2 Equations}\label{sec:4}

It is known that in the ABS list  Q1($\Delta$) and Q2
can be obtained from the degeneration of Q3($\Delta$) equation \cite{Atk-JPA-2008},
and the degeneration has been used to find soliton and rational solutions
for Q1($\Delta$) and Q2 equations \cite{NAH-JPA-2009,ZZ-JNMP-2018}.
In what follows, we derive elliptic solitons for Q1($\Delta$) and Q2
via implementing degenerations on  Q3($\Delta$) and its solutions.

\subsection{Degeneration to Q1(\texorpdfstring{$\Delta$}{Delta})}\label{sec:4-1}

We next derive solutions of the Q1($\Delta$) equation as the coalescence limit
\begin{equation*}
q=k+\varepsilon, \quad \varepsilon \to 0,
\end{equation*}
of the leKdV-II($\delta$) equation \eqref{eq:II-delta}.
Here the coalescence is taken with respect to the elliptic parameters $k$ and $q$.

First, let us work on the plane wave solution \eqref{eq:Yyij-e},
from which we exhibit the degeneration in detail.
Expanding the functions $y_{ij,[e]}$ defined in \eqref{eq:yij-e} in powers of $\varepsilon$ yields
\begin{equation}
y_{ij,[e]}=[\varphi_{i,[e]}(k),~ \varphi_{j,[e]}(k)]^{-}_{[0],[n]}
+\varepsilon \, x_{ij,[e]}
+O(\varepsilon^2),
\end{equation}
where the first-order coefficient is
\begin{equation}
x_{ij,[e]}
=[\varphi_{i,[e]}(k), ~\partial_k  \varphi_{j,[e]}(k)]^{-}_{[0],[n]}.
\end{equation}
Noting that
\begin{equation}\label{F}
\begin{split}
&\partial_k \varphi_{i,[e]}(k) ={F}_i \, \varphi_{i,[e]}(k),
\qquad i=1,2,\\
 &   {F}_1=n(\zeta(k+a)-\zeta(k))+m(\zeta(k+b)-\zeta(k))+\zeta(k-\xi)-\zeta(k),\\
&    {F}_2=n(\zeta(k-a)-\zeta(k))+m(\zeta(k-b)-\zeta(k))+\zeta(k+\xi)-\zeta(k),
  \end{split}
\end{equation}
where $\xi=na+mb+\xi_0$ is defined  in \eqref{eq:u0e},
the coefficients $x_{ij,[e]}$ are obtained explicitly below
\begin{subequations}\label{xij-e}
\begin{align}
x_{11,[e]}&=\varphi_{1,[e]}(k)\widetilde{\varphi}_{1,[e]}(k)\chi_{a+k,-k}(\xi-k),\\
x_{12,[e]}&=\varphi_{1,[e]}(k)\widetilde{\varphi}_{2,[e]}(k)\widetilde{F}_2
-\widetilde{\varphi}_{1,[e]}(k)\varphi_{2,[e]}(k) {F}_2,\\
x_{21,[e]}&=\varphi_{2,[e]}(k)\widetilde{\varphi}_{1,[e]}(k)\widetilde{F}_1
-\widetilde{\varphi}_{2,[e]}(k)\varphi_{1,[e]}(k) {F}_1,\\
x_{22,[e]}&=-\varphi_{2,[e]}(k)\widetilde{\varphi}_{2,[e]}(k)\chi_{a-k,k}(\xi+k),
\end{align}
\end{subequations}
where $\chi_{u,v}(z)$ is defined  as in \eqref{eq:chi}.
Under the condition $B=C$, the zeroth-order term  in the expansion of $Y$ vanishes, yielding
\begin{equation}
Y
=
\varepsilon \left( Ax_{11,[e]}
+B(x_{12,[e]}+x_{21,[e]})+Dx_{22,[e]}\right)
+O(\varepsilon^2).
\end{equation}
In addition, noticing the expansions
\begin{align*}
\delta  \to \wp'(k)^{2}(AD-B^{2})+O(\varepsilon),\qquad
 (\wp(k)-\wp(q))^2\to \varepsilon^2\wp'(k)^2+O(\varepsilon^3)
\end{align*}
and introducing function
\begin{equation}\label{Z-e}
Z=s^{-2n}t^{-2m} \left( Ax_{11,[e]}
+B(x_{12,[e]}+x_{21,[e]})+Dx_{22,[e]}\right),
\end{equation}
 the leKdV-II($\delta$) equation  \eqref{eq:II-delta} reduces, in the limit $\varepsilon\to0$,   to the Q1($\Delta$) equation
\begin{equation}\label{eq:Q1delta}
s^{2}(Z-\widetilde{Z})(\widehat{Z}-\widehat{\widetilde{Z}})
-
t^{2}(Z-\widehat{Z})(\widetilde{Z}-\widehat{\widetilde{Z}})
=
\Delta
\left(
\frac{1}{t^{2}}-\frac{1}{s^{2}}
\right),
\end{equation}
where
\begin{equation}\label{eq:Q1delta-value}
\Delta=\wp'(k)^4 (AD-B^2).
\end{equation}
Thus, equation \eqref{Z-e} together with \eqref{xij-e} provides the plane wave solutions to the
 Q1($\Delta$) equation \eqref{eq:Q1delta}.

In the general case, we can conduct the above limiting procedure
on the elliptic $N$-soliton solution  given in  \eqref{eq:YN},
and obtain explicit elliptic
$N$-soliton solution for the Q1($\Delta$) equation.
We especially note that the denominator in $\varphi_j[N](k)$ is independent of $k$, because the
last row is actually removed by the column $\boldsymbol{\sigma}$.
Differentiating the numerator with respect to $k$ amounts to replacing the last row
by its derivative with respect to $k$, which leads to the following theorem.
\begin{theorem}
We introduce
\begin{equation}\label{eq:partial-varphi-N}
\bm{\mathcal{F}}_j[N](k):= \partial_k \varphi_j[N](k)
=
\frac{
\left|\widehat{N}(\bm{\mathcal{F}}_j)\right|_{[n]}
}{
\left|
\widehat{N-1}(\bm{\mathcal{F}}_j),
~\boldsymbol{\sigma}
\right|_{[n]}
},  \qquad j=1,2
\end{equation}
where
\begin{equation}
\bm{\mathcal{F}}_j
=
\left(
\phi_{1,[e]},
\phi_{2,[e]},
\ldots,
\phi_{N,[e]},
{F}_j\varphi_{j,[e]}(k)
\right)^T.
\end{equation}
Consequently, the Q1($\Delta$) equation \eqref{eq:Q1delta} with
$\Delta$ defined by \eqref{eq:Q1delta-value},
admits the following $N$-soliton solution:
\begin{equation}
Z[N]
=
s^{-2n}t^{-2m} \mathcal N_N(k)^{-2}
\left(
Ax_{11}[N]
+
B\bigl(x_{12}[N]+x_{21}[N]\bigr)
+
Dx_{22}[N]
\right),
\end{equation}
where $\mathcal N_N(k)$ is defined in \eqref{eq:factor}, and
\begin{equation}\label{eq:xijN}
x_{ij}[N]
=
\varphi_i[N](k)
\widetilde{\bm{\mathcal{F}}}_j[N](k)
-
\widetilde{\varphi}_i[N](k)
\bm{\mathcal{F}}_j[N](k),
\qquad i,j=1,2.
\end{equation}

\end{theorem}

\subsection{Degeneration to  Q2}\label{sec:4-2}

We derive the Q2 equation by taking the  limit
$q=k+\varepsilon$, $\varepsilon\to 0$,  in the Q3($\Delta$) equation \eqref{eq:Q3delta},
where we impose the normalization $AD-BC=1$
in $\Delta$, i.e., (cf.\eqref{eq:Q3-deltavalue} and \eqref{eq:II-deltavalue-g})
\begin{equation*}
  \Delta=(\wp(k)-\wp(q))^2 \rho(k)\rho(q).
\end{equation*}
The arguments involved in the Q3($\Delta$) equation \eqref{eq:Q3delta}  have the following
 expansions in terms of $\varepsilon$:
\begin{subequations}\label{eq:series}
\begin{align}
&\rho(q)\to  \rho(k)+\varepsilon \rho_k(k)+\frac{ \varepsilon^2}{2} \rho_{kk}(k)+O(\varepsilon^3), \\
&\mu \to s- \varepsilon \frac{\wp'(k)}{2s}
     -\varepsilon^2\left(\frac{\wp''(k)}{4s}+\frac{\wp'(k)^2}{8 s^3}\right)\nonumber\\
    &\qquad-\varepsilon^3\left(\frac{\wp^{(3)}(k)}{12 s}+\frac{\wp'(k) \wp''(k)}{8 s^3}+\frac{\wp'(k)^3}{16 s^5}\right)\nonumber\\
    &\qquad-\varepsilon^4\left(\frac{\wp^{(4)}(k)}{48 s}+\frac{\wp^{(3)}(k) \wp'(k)}{24 s^3}+\frac{\wp''(k)^2}{32 s^3}+\frac{3 \wp'(k)^2 \wp''(k)}{32 s^5}+\frac{5 \wp'(k)^4}{128 s^7}\right)+O(\varepsilon^5),\\
&\frac{(\wp(k)-\wp(q))^2}{st\mu\nu}\to  \varepsilon^2\frac{\wp'(k)^2}{s^2t^2}
           +\varepsilon^3 \left(\frac{\wp'(k)\wp''(k)}{s^2t^2}+\frac{(s^2+t^2)\wp'(k)^3}{2 s^4 t^4}\right)\nonumber \\
          &+\varepsilon^4\left(\frac{\wp''(k)^2}{4 s^2 t^2}+\frac{\wp^{(3)}(k) \wp'(k)}{3 s^2 t^2}+\frac{3 \wp'(k)^2 \wp''(k)(s^2+t^2)+\wp'(k)^4}{4 s^4 t^4}
         +\frac{3 \wp'(k)^4(s^4+t^4)}{8 s^6 t^6}\right)+O(\varepsilon^5),\\
&\mu^{n}\nu^{m}\to \Lambda^{\frac12}
    +\frac{\varepsilon}{2}\Lambda_k\Lambda^{-\frac12}
    +\varepsilon^2\left(\frac14\Lambda_{kk}\Lambda^{-\frac12}-\frac18\Lambda_k^2\Lambda^{-\frac32}\right)
    +O(\varepsilon^3).
\end{align}
\end{subequations}
Here
\begin{equation}\label{Lambda}
\Lambda=s^{2n}t^{2m}=(\wp(a)-\wp(k))^n(\wp(b)-\wp(k))^m
\end{equation}
and $\Lambda_k=\partial_k \Lambda,~ \Lambda_{kk}=\partial^2_k \Lambda$.
We assume that the solution $X$   of the Q3($\Delta$) equation admits the following expansion\footnote{
This form is reasonable due to the definition of $y_{ij}$ in \eqref{eq:yij}
and the discrete Abel formula \eqref{eq:Abel-1}.}
\begin{equation}\label{eq:X-expansion}
X
=
2\rho(k)+\varepsilon \rho_k(k)+\varepsilon^2\rho(k)\,\Xi
+O(\varepsilon^3),
\end{equation}
where $\Xi=\Xi(n,m;s,t)$.
Substituting the expansions \eqref{eq:series} into the Q3($\Delta$) equation \eqref{eq:Q3delta}
 and collecting the first nonzero term,  which occurs at order  $\varepsilon^4$, yields
\begin{equation}
\begin{aligned}
&\frac{\wp'(k)^2}{4}
\left(
\frac1{t^2}-\frac1{s^2}
\right)
\left(\Xi+\widetilde{\Xi}+\widehat{\Xi}+\widetilde{\widehat{\Xi}}+\frac{\rho_k(k)^2}{\rho(k)^2}-2\frac{\rho_{kk}(k)}{\rho(k)}\right)\\
&=-s^2(\Xi-\widetilde{\Xi})(\widehat{\Xi}-\widehat{\widetilde{\Xi}})
+t^2(\Xi-\widehat{\Xi})(\widetilde{\Xi}-\widehat{\widetilde{\Xi}})
+\frac{\wp'(k)^4}{16}
\frac{s^6-2s^4t^2+2s^2t^4-t^6}{s^6t^6}.
\end{aligned}
\end{equation}
We thus obtain the  standard form of the Q2 equation
\begin{equation}\label{eq:Q2}
\mathcal{P}(W-\widehat{W})(\widetilde{W}-\widehat{\widetilde{W}})
-\mathcal{Q}(W-\widetilde{W})(\widehat{W}-\widehat{\widetilde{W}})
+\mathcal{PQ}(\mathcal{P-Q})\bigl(W+\widetilde{W}+\widehat{W}+\widehat{\widetilde{W}}
-\mathcal{P}^2+\mathcal{PQ}-\mathcal{Q}^2\bigr)=0
\end{equation}
via the transformation
\begin{equation}\label{W}
W=\frac{1}{\wp'(k)^2}\left(4\Xi+\frac{\rho_k(k)^2}{\rho(k)^2}-2\frac{\rho_{kk}(k)}{\rho(k)}\right), \qquad  \mathcal{P}=\frac{1}{s^2}, \qquad \mathcal{Q}=\frac{1}{t^2}.
\end{equation}

To derive explicit solutions,
the expansion of  $X$ in \eqref{eq:X-expansion}, together with the relation \eqref{eq:X-Y} between  $X$ and $Y$, suggests the following expansion for $Y$
 \begin{equation}\label{eq:series-Y-1}
Y
=
2\rho(k)\Lambda
+\varepsilon(\rho(k)\Lambda)_k
+\varepsilon^2
\left(
\frac{\rho(k)}{2}\Lambda_{kk}+\frac{\rho_k(k)}{2}\Lambda_{k}
-\frac{\rho(k)\Lambda_k^2}{4\Lambda}
+\rho(k)\Lambda\Xi
\right)
+O(\varepsilon^3),
\end{equation}
Setting $B=-C$, $A=\varepsilon\mathcal{A}$ and $D=\varepsilon\mathcal{D}$,
the normalization $AD-BC=1$ then gives
\begin{equation}
B=1-\frac{\mathcal{AD}}{2}\varepsilon^2+O(\varepsilon^4).
\end{equation}
For the case of plane wave solution,
expanding the eigenfunction representation  \eqref{eq:Yyij-e} of $Y$  up to $O(\varepsilon^2)$
and comparing the result with
\eqref{eq:series-Y-1} yield the representation of $\Xi$
\begin{align}\label{4.22}
\Xi
=&
\frac{\Lambda_k^2}{4\Lambda^{2}}
+\frac{\rho_k(k)\Lambda_k}{2\rho(k)\Lambda}
+\frac{\rho_{kk}(k)}{2\rho(k)}
-\mathcal{AD}\nonumber\\
+&\frac{
\mathcal{A}x_{11,[e]}
+\mathcal{D}x_{22,[e]}
+\widetilde{{F}}_1 {F}_2\widetilde{\varphi}_{1,[e]}
\varphi_{2,[e]}
-{F}_1\widetilde{{F}}_2\varphi_{1,[e]}\widetilde{\varphi}_{2,[e]}
}{\rho(k)\Lambda},
\end{align}
where $x_{ij,[e]}$ are defined in \eqref{xij-e}, $F_1$ and $F_2$ are defined in \eqref{F},
and $\rho(k)=-\wp'(k)$.
Substituting \eqref{4.22} into \eqref{W},
we can obtain the plane wave solution for the Q2 equation \eqref{eq:Q2}.

We now derive the elliptic $N$-soliton solution of the Q2 equation.
We start from the unnormalized solution of the Q3($\Delta$) equation:
\begin{equation*}
X[N]=
s^{-n}\mu^{-n}t^{-m}\nu^{-m}(Ay_{11}[N]+By_{12}[N]+Cy_{21}[N]+Dy_{22}[N]),
\end{equation*}
with $\Delta$ replaced by
 \begin{equation*}
\Delta_N=
\bigl(\wp(k)-\wp(q)\bigr)^2
\rho_N(k)\rho_N(q)(AD-BC),
\end{equation*}
where $y_{ij}[N]$ and  $\rho_N(k)$ are given by \eqref{eq:yijN} and \eqref{eq:rho-N}  respectively.
Repeating the preceding analysis with the replacements
$x_{ij}\to x_{ij}[N]$,
$\rho(k)\to \rho_N(k)$,
$F_i\varphi_{i,[e]}\to \bm{\mathcal{F}}_i[N](k) $,
we  obtain an explicit expression for  $\Xi[N]$ in terms of  the Darboux-transformed eigenfunctions.
The resulting elliptic $N$-soliton solution of the
Q2 equation is stated in the following theorem.

\begin{theorem}\label{thm:Q2-N}
  The Q2 equation \eqref{eq:Q2} admits the elliptic $N$-soliton
  solution
  \begin{equation}\label{eq:Q2-N-solution}
  \begin{split}
    W[N]&=\frac{1}{\wp'(k)^2}\left( \left(\frac{\Lambda_k}{\Lambda}+\frac{\rho_{N,k}(k)}{\rho_N(k)}\right)^2
    -4\mathcal{AD}\right.\\
    &+\left.\frac{4}{\rho_N(k)\Lambda}
    \left(\mathcal{A}x_{11}[N]
+\mathcal{D}x_{22}[N]
+\widetilde{\bm{\mathcal{F}}}_1[N](k)\bm{\mathcal{F}}_2[N](k)
-\bm{\mathcal{F}}_1[N](k)\widetilde{\bm{\mathcal{F}}}_2[N](k)\right)\right),
\end{split}
  \end{equation}
where $\rho_{N}(k)$,  $\bm{\mathcal{F}}_j[N](k)$, $x_{ij}[N]$,  $i,j=1,2$, and $\Lambda$,
are defined by \eqref{eq:rho-N}, \eqref{eq:partial-varphi-N}, \eqref{eq:xijN} and  \eqref{Lambda},
respectively.
\end{theorem}

Note that Q2 is special in the ABS list in the sense that its bilinear form is not found so far
(cf.\cite{WZM-PD-2024}). We hope the above Casoratian formulation of its solutions
will be helpful in finding a bilinearization of this equation.

\section{Conclusions}\label{sec:5}

In this paper we have shown that the known Darboux transformation
of the discrete Schr\"odinger spectral problem \eqref{eq:laxpair-H1-A}
is valid as well for \eqref{eq:laxpair-H1-B} and \eqref{eq:laxpair-H1-C}.
Thus we achieved a full Darboux transformation for the Lax triplet
\eqref{eq:laxpair-H1} of the lpKdV equation.
As a result, we obtained elliptic solitons for the lpKdV equation
and also explicit Casoratian formulae for the corresponding eigenfunction equations.
This enables us to present elliptic solitons in Casoratian form for the
eigenfunction equations
of the lattice KdV type, namely,
H3($0$), H3($\delta$), Q1($\Delta$), Q2 and Q3($\Delta$) in the ABS list.
In the construction of these eigenfunction equations, the discrete Abel formula \eqref{eq:Abel-1}
plays a key role, in which an `integration' constant is connected to the parameter $\delta$
(or $\Delta$). In this paper, we have determined the value of this constant
for the (elliptic) plane wave case and the elliptic soliton case.
This allowed us to formulate the eigenfunction equations together with
their parameters and elliptic plane wave solutions and elliptic $N$-solitons.
Solutions of Q1($\Delta$) and Q2 were obtained as degeneration of  Q3($\Delta$).
We hope the research of this paper will be helpful in understanding eigenfunction equations
and their various solutions related to higher order spectral problems.

\vskip 20pt

\section*{Acknowledgments}
	The work of XL is supported by the NSFC Grant (12301309) and was also supported by the China Scholarship Council.
	The work of YS is supported by the NSFC Grant (12426613).
	The work of DJZ is supported by the NSFC Grant (12271334, 12671301) and the NSFC Bilateral Program (1241540016).

\vskip 20pt

\appendix
\section{Determinant identities and elliptic function formulae}
\label{app:det-elliptic}

We first recall the Desnanot--Jacobi identity\cite[Theorem~2.6]{Gragg-SIAM-1972} (also known as
Dodgson condensation).
We state it in a notation that distinguishes a matrix from the determinants of its minors.

\begin{theorem}[Desnanot-Jacobi identity]
\label{lemma-DJ}
Let $\mathcal{M}=(m_{ij})$ be an $(N+3)\times(N+3)$ matrix and
$M=\det \mathcal{M}$.  Denote by $M_i^{\,j}$ the determinant obtained by
deleting the $i$th  row and  $j$th column from $\mathcal{M}$, and by
$M_{i,k}^{\,j,l}$ the determinant obtained by deleting rows $i,k$ and
columns $j,l$.  Then
\begin{equation}
M_{N+1,N+2}^{\,1,N+2}\,M
=
M_{N+1}^{\,1}\,M_{N+2}^{\,N+2}
-
M_{N+1}^{\,N+2}\,M_{N+2}^{\,1}.
\label{eq:DJ}
\end{equation}
\end{theorem}

Let $\sigma(z)$, $\zeta(z)$, and $\wp(z)$ be the Weierstrass $\sigma$, $\zeta$, and
elliptic $\wp$ functions of $z$  associated with the lattice $2w\mathbb Z+2w'\mathbb Z$.
They are related by
\begin{equation} \label{eq:wpp}
 \zeta(z)=\frac{\sigma'(z)}{\sigma(z)},
 \qquad
 \wp(z)=-\zeta'(z),
 \qquad \wp'(z)
 =-\frac{\sigma(2z)}{\sigma^4(z)}.
\end{equation}
The point $(\wp(z),\wp'(z))$ lies on the standard Weierstrass curve
\begin{align}\label{eq:ell-cur}
 y^2
 &=4x^3-g_2x-g_3
 =4(x-e)(x-e')(x-e''),
\end{align}
where
\begin{equation}
 e=\wp(w),\qquad e'=\wp(w'),\qquad
 e''=\wp(w+w')=-e-e',
\end{equation}
and
\begin{equation}
 g_2=-4(ee'+e'e''+e''e),
 \qquad
 g_3=4ee'e''.
\end{equation}
The zeros of $\wp'(z)$ are the  three  half-periods $w$, $w'$, and $w+w'$  modulo periods.
For $\omega\in\{w,w'\}$, the quasi-periodicity relations are
\begin{align}
 \zeta(z+2\omega)&=\zeta(z)+2\zeta(\omega),\\
 \sigma(z+2\omega)&=-e^{2\zeta(\omega)(z+\omega)}\sigma(z).
\end{align}

The Weierstrass $\sigma$ function satisfies the addition formula
\begin{align}\label{eq:add-5}
 &\sigma(x+a)\sigma(x-a)\sigma(y+b)\sigma(y-b)
 -\sigma(x+b)\sigma(x-b)\sigma(y+a)\sigma(y-a)
 \nonumber\\
 &\hspace{20mm}
 =\sigma(x+y)\sigma(x-y)\sigma(a+b)\sigma(a-b).
\end{align}
Another identity involving the $\sigma$ function is
\begin{align}
 \chi_{u,v}(z)
 &:=\zeta(u)+\zeta(v)+\zeta(z)-\zeta(u+v+z)
 \nonumber\\
 &=\frac{\sigma(u+v)\sigma(u+z)\sigma(v+z)}
 {\sigma(u)\sigma(v)\sigma(z)\sigma(u+v+z)}.
 \label{eq:chi}
\end{align}
The combination $\chi_{u,v}(z)$ is symmetric in $u,v,z$ and elliptic in each argument.
A related two-variable combination is
\begin{equation} \label{eq:add-2}
\eta_u(z)
 :=\zeta(z+u)-\zeta(z)-\zeta(u)
 =\frac{1}{2}\frac{\wp'(z)-\wp'(u)}{\wp(z)-\wp(u)},
\end{equation}
which satisfies  the relation
\begin{equation} \label{eq:add-3}
 \eta_u^2(z)
 =\wp(z)+\wp(u)+\wp(z+u).
\end{equation}

We next introduce the Lam\'e function
\begin{equation}
 \Psi_x(k)=\Phi_x(k)e^{-\zeta(k)x},  \qquad
  \text{with}~~ \Phi_x(k)=\frac{\sigma(x+k)}{\sigma(x)\sigma(k)},
\end{equation}
which has the symmetry relation
\begin{equation}
 \Psi_x(-k)=-\Psi_{-x}(k),
\end{equation}
by the oddness of $\sigma$ and $\zeta$.
$\Phi_x(k)$ satisfies
\begin{equation}\label{eq:add-1}
 \Phi_z(u)\Phi_z(-u)=\wp(z)-\wp(u).
\end{equation}
We illustrate the half-period identities with $w$, for which $e=\wp(w)$.
Quasi-periodicity gives
\begin{equation}
 \Phi_x(-w)=e^{-2\zeta(w)x}\Phi_x(w).
\end{equation}
Consequently, we have
 \begin{equation}
\Psi_x(w)=\Psi_x(-w)=-\Psi_{-x}(w),
\end{equation}
and
 \begin{equation}\label{eq:Lame-half-period}
 \bigl(\Psi_x(w)\bigr)^2
 =\Phi_x(w)\Phi_x(-w)
 =\wp(x)-e,
\end{equation}
with the aid of \eqref{eq:add-1}.
Similarly, evaluating \eqref{eq:add-2} at $z=w$ gives
\begin{equation}
 4\bigl(\eta_x(w)\bigr)^2\bigl(\wp(x)-e\bigr)^2
 =\wp'(x)^2.
\end{equation}
Combining this relation with the standard Weierstrass curve \eqref{eq:ell-cur}, we  obtain the parametrization
\begin{align*}
&(Y,X)=(\eta_x(w),\Psi_x(w))
\end{align*}
 of the elliptic curve
\begin{equation}\label{eq:Lame-curve}
Y^2=X^2+3e+gX^{-2},  \qquad g=(e-e')(e-e'')=3e^2-\frac{g_2}{4},
\end{equation}
which has been used in \cite{FN-SAPM-2022,NP-JNMP-2003,SZN-JMP-2017,NZZ-TMP-2025}.
Elliptic parameter curves also arise naturally in lattice integrable systems,
a prominent example being the $Q4$ equation\cite{HJN-book-2016}.

For  the parametrization relevant to the $Q3(\Delta)$ equation used in the paper, we consider  the Jacobi quartic form
\begin{equation}\label{eq:Q3-ell-curve}
\Gamma_{\mathbf{k},\mathbf{q}}
 =
 \left\{
 (x,y)\in\mathbb C^2:
 y^2=(x^2-\mathbf{k}^2)(x^2-\mathbf{q}^2)
 \right\},
\end{equation}
with  branch points  $x=\pm\mathbf{k}$ and $x=\pm\mathbf{q}$.
The parametrization  in Theorem~\ref{thm:Q3}
gives
\begin{equation}\label{eq:Q3-ell-points}
 P^2
 =
 (\alpha^2-\mathbf{k}^2)
 (\alpha^2-\mathbf{q}^2),
 \qquad
 Q^2
 =
 (\beta^2-\mathbf{k}^2)
 (\beta^2-\mathbf{q}^2),
\end{equation}
and hence
\[
 (\alpha,P),\;(\beta,Q)
 \in \Gamma_{\mathbf{k},\mathbf{q}}.
\]

\vskip 20pt

\bibliographystyle{amsalpha}

\end{document}